\documentclass[11pt]{article}

\usepackage{amssymb}
\usepackage{amsmath}

\usepackage[T1]{fontenc}
\usepackage{concrete,eulervm}

\usepackage{sectsty}
\allsectionsfont{\mdseries}

\usepackage{a4}
\usepackage{amsthm}
\usepackage{parskip}

\usepackage[round]{natbib}

\usepackage{hyperref}
\usepackage{xcolor}
\definecolor{dark-red}{rgb}{0.4,0.15,0.15}
\definecolor{dark-blue}{rgb}{0.15,0.15,0.4}
\definecolor{medium-blue}{rgb}{0,0,0.5}
\hypersetup{
    colorlinks, linkcolor={dark-red},
    citecolor={dark-blue}, urlcolor={medium-blue}
}

\newtheorem{theorem}{Theorem}

\newtheorem{lemma}{Lemma}

\title{Strict Comparisons of Infinite Utility Streams\footnote{We thank Marcus Pivato for helpful discussion.}}
\author{Michael Greinecker\thanks{ENS Paris-Saclay, CEPS, \href{mailto:michael.greinecker@ens-paris-saclay.fr}{michael.greinecker@ens-paris-saclay.fr}} and Michael Nielsen\thanks{Texas A \& M University, 
Department of Philosophy, \href{mailto:michael.nielsen@tamu.edu}{michael.nielsen@tamu.edu}}}

\begin{document}
\maketitle
\begin{abstract}
There exists a preference relation on infinite utility streams that does not discriminate between different periods, satisfies the Pareto criterion, and so that almost all pairs of utility streams are strictly comparable. Such a preference relation provides a counterexample to a claim in [Zame, William R. ``Can intergenerational equity be operationalized?'' Theoretical Economics 2.2 (2007): 187-202.] 
\end{abstract}

\section*{Introduction}

Suppose there is an infinite sequence of generations, living neatly one after another. If we solve the utilitarian aggregation problem within each generation, we are left with the problem of how to compare the resulting utility streams.\footnote{\citet{PivatoFleurbaey2024} survey the state of the art.} If we do not want to discriminate arbitrarily between different generations, our social preferences should  be unaffected by swapping the utilities of two generations (and thus, any finite number of generations). If we are sensitive to the plight of every generation, our social preferences should satisfy a variant of the Pareto criterion.  

\citet{Svensson1980} proved the existence of a (complete and transitive) preference relation on infinite utility streams that satisfies the Pareto criterion and that is invariant under switching the utility levels of finitely many generations. Such a preference relation must be an unruly object. \citet{Diamond1965} has shown that a preference relation on infinite utility streams that satisfies the Pareto criterion and that is invariant under switching the utility levels of finitely many generation cannot be continuous in the topology of uniform convergence, and \citet{BasuMitra2003} have shown that it cannot be represented by a utility function. \citet{Zame2007}, \citet{Lauwers2010}, \citet{Dubey2011} and \cite{DubeyLaguzzi2023} have shown that such preferences cannot be proven to exist if one weakens the axiom of choice to the axiom of dependent choice, which suffices for almost everything the working economist has to deal with.\footnote{\citet{Lauwers2010} and \citet{Dubey2011} show that the existence of such preferences implies the existence of a non-Ramsey set, \citet{Zame2007}, shows that the existence of such preferences implies the existence of a non-measurable set, and \cite{DubeyLaguzzi2023} show that the existence of such preferences implies the existence of a non-Baire set. By results of \citet{Shelah1984}, it is relatively consistent that dependent choice holds and all sets of reals are Baire sets, while the consistency of dependent choice and all sets of reals being measurable holds only relatively to an inaccessible cardinal. The consistency strength of having dependent choice and only Ramsey sets remains unknown.} That one cannot prove the existence of such preferences without heavy-handed use of the axiom of choice has previously been conjectured by \citet{FleurbaeyMichel2003}.

\citet{Zame2007} also claimed (in his Theorem 1 and Theorem 1') that an irreflexive strict preference relation that is invariant under switching the utility levels of finitely many generations must be extremely indecisive in that the set of pairs of utility streams that cannot be strictly compared must have full outer measure. This is not so. We show that preferences on infinite utility streams can be quite decisive, even if we require them to satisfy the strict Pareto criterion. For the preferences we construct, the set of pairs of utility streams that cannot be strictly compared (in our case, the graph of the indifference relation) is measurable and has measure zero. To ``construct'' such preferences, a tiny modification of the original argument of \citet{Svensson1980} suffices.

\section*{Preliminaries and Main Result}
For two sequences $x=(x_n)$ and $y=(y_n)$ in $\mathbb{R}^\mathbb{N}$, we write $x\geq y$ if $x_n\geq y_n$ for all $n\in\mathbb{N}$, $x>y$ if $x\geq y$ and $x\neq y$, and $x\gg y$ if $x_n>y_n$ for all $n\in\mathbb{N}$. For a generic (binary) relation $\succeq$, we denote its asymmetric part by $\succ$ and its symmetric part by $\sim$. A relation $\succeq$ on a subset of $\mathbb{R}^\mathbb{N}$ satisfies \emph{strict Pareto} if $x>y$ implies $x\succ y$, and satisfies \emph{weak Pareto} if $x\gg y$ implies $x\succ y$. 

A \emph{finite permutation} is a permutation (bijection) $\sigma:\mathbb{N}\to\mathbb{N}$ such that $n=\sigma(n)$ for all but finitely many $n\in\mathbb{N}$. We denote the set of finite permutations by $\mathbb{F}$. Under the operation of composition, $\mathbb{F}$ is a countable group. If $x\in\mathbb{R}^\mathbb{N}$ and $\sigma\in\mathbb{F}$, we write, slightly abusing notation, $\sigma(x)$ for the sequence $(x_{\sigma(n)})$. A relation $\succeq$ on $X=[0,1]^\mathbb{N}$ or $X=\{0,1\}^\mathbb{N}$ satisfies \emph{intergenerational equity} if $x\succ y$ holds if and only if $\sigma(x)\succ\tau(y)$ for all $x,y\in X$ and $\sigma,\tau\in\mathbb{F}$. The following characterization of intergenerational equity for well-behaved preferences (already noted by \citet{Zame2007}) is straightforward to prove.

\begin{lemma}Let $\succeq$ be a complete and transitive relation on $X=[0,1]^\mathbb{N}$ or $X=\{0,1\}^\mathbb{N}$. Then $\succeq$ satisfies intergenerational equity if and only if $x\sim\sigma(x)$ holds for all $x\in X$ and $\sigma\in\mathbb{F}$.
\end{lemma}

We endow  $X=[0,1]^\mathbb{N}$ and $X=\{0,1\}^\mathbb{N}$ with their natural product topologies and Borel $\sigma$-algebras. Either space is Polish (separable and completely metrizable). We endow (the Borel sets of) $X=[0,1]^\mathbb{N}$ with the countable product of Lebesgue measure, and (the Borel sets of) $X=\{0,1\}^\mathbb{N}$ with the countable fair coin-flipping measure. In either case, we endow (the Borel sets of) $X\times X$ with the resulting product measures. These product measures are atomless, which in the present context is equivalent to no point having strictly positive mass. A property holds for \emph{almost all} pairs in $X\times X$ if it holds outside a Borel set of measure zero. The outer measure of an arbitrary subset of $X\times X$ is the infimum of the measures of all Borel supersets. The inner measure of an arbitrary subset of $X\times X$ is the supremum of the measures of all Borel subsets. 

We now have all the ingredients for our result.

\begin{theorem}\label{main}There exist complete and transitive relations on $X=[0,1]^\mathbb{N}$ and $X=\{0,1\}^\mathbb{N}$ that satisfy strict Pareto, intergenerational equity, and such that almost all pairs in $X\times X$ are strictly comparable.
\end{theorem}
\begin{proof}
Write $x \equiv y$ if there exists a finite permutation $\sigma$ such that $\sigma(x)=y$. The group structure on $\mathbb{F}$ implies that $\equiv$ is an equivalence relation. By \citet[Theorem 6.4]{Kechris1995}, the graph of a Borel measurable function between two Polish spaces is measurable. By Fubini's theorem, it has measure zero under any product probability measure with atomless marginals. Consequently, the graph of $\equiv$ can be written as a countable union of measure zero sets as
\[\bigcup_{\sigma\in\mathbb{F}}\Big\{\big(x,\sigma(x)\big)\mid x\in X\Big\}\]
and is, therefore, a measure zero set itself.
We define a relation $\geqq$ on the partition into equivalence classes $X/\equiv$ such that for two equivalence classes, $[x]$ and $[y]$, we have $[x] \geqq [y]$ exactly when $x\geq \sigma(y)$ holds for some finite permutation $\sigma$. Then, $\geqq$ is easily seen to be a partial order. By the Szpilrajn extension theorem, \citet{Szpilrajn1930}, $\geqq$ extends to a linear order on $X/\equiv$. We now define the desired relation $\succeq$ on $X$ so that $x\succeq y$ holds exactly when $[x] \geqq [y]$. The set of pairs that are not strictly comparable is then exactly $\equiv$, a set of measure zero.
\end{proof}

The ``construction'' we used in the proof of Theorem \ref{main} is virtually identical to the one used by \citet{Svensson1980}. The only difference is that we use an extension from partial orders to linear orders, while Svensson uses an extension from preorders to complete preorders that preserves strict preferences. The latter kind of extension might create additional indifferences, the former does not.

\section*{Discussion} 

The choice of the domains $X=[0,1]^\mathbb{N}$ and $X=\{0,1\}^\mathbb{N}$ follows the literature but is inessential to the argument. The argument shows that for every Borel set $S\subseteq\mathbb{R}$, there exists a complete and transitive relations on $S^\mathbb{N}$ that satisfies strict Pareto and intergenerational equity. If we put on (the Borel sets of) $S$ any Borel probability measure that is not simply a point mass, the corresponding product measures will be atomless, and the graph of the indifference relation will have measure zero.

\citet{Zame2007} uses his Theorem 1 (or 1') to prove the remaining theorems in his paper. However, this does not cause any problems for the remaining theorems. In the proof of his Theorem 1, Zame starts with a strict relation $\succ$ on $X$ that satisfies intergenerational equity and proves that the sets
\[\big\{(x,y)\in X\times X\mid x\succ y\big\}\textnormal{ and }\big\{(x,y)\in X\times X\mid y\succ s\big\}\]
both have inner measure zero, and concludes that the set 
\[\big\{(x,y)\in X\times X\mid x\succ y\big\}\cup\big\{(x,y)\in X\times X\mid y\succ s\big\}\]
has inner measure zero, too. This last step is not warranted; inner measure is superadditive but, in general, not subadditive. But the step is warranted when these two sets are measurable, as they are when Zame uses his Theorem 1 (or 1') to prove the remaining theorems.


\bibliography{References}

\end{document}